\documentclass[runningheads,envcountsame]{llncs}
\usepackage[utf8]{inputenc}
\usepackage{color}
\usepackage{nicefrac}
\usepackage{enumitem}
\usepackage{url}
\usepackage{amsmath,amsfonts}
\usepackage{graphicx}
\usepackage{amssymb}
\usepackage{pifont}
\usepackage[flushleft]{threeparttable}

\newcommand{\cmark}{\ding{51}}
\newcommand{\xmark}{\ding{55}}

\DeclareMathOperator*{\argmin}{arg\,min}
\newtheorem{observation}{Observation}
\newtheorem{constraint}{Constraint}

\title{Iterative Deliberation via Metric Aggregation}
\author{Gil Ben Zvi\inst{2} \and Eyal Leizerovich\inst{1} \and Nimrod Talmon\inst{1}}
\institute{Ben-Gurion University \and NRGene LTD}
\begin{document}
\maketitle

\begin{abstract}
We investigate an iterative deliberation process for an agent community wishing to make a joint decision.
We develop a general model consisting of a community of $n$ agents, each with their initial ideal point in some metric space $(X, d)$, such that in each iteration of the iterative deliberation process, all agents move slightly closer to the current winner, according to some voting rule $\mathcal{R}$.
For several natural metric spaces and suitable voting rules for them, we identify conditions under which such an iterative deliberation process is guaranteed to converge.
\end{abstract}

\section{Introduction}\label{section:introduction}

Agent communities wishing to reach joint decisions usually get involved in some voting process:
  agent preferences wrt.\ some agreed-upon options are being elicited, and their preferences are being aggregated through the use of some aggregation method.
Correspondingly, much of the research in computational social choice~\cite{moulin2016handbook} evolves around such aggregation methods, usually referred to as voting rules.

If voter preferences are rather diverse, then using a voting rule in a straightforward way might mean that the aggregated result (i.e., the result of the election) is not well-accepted by the agent community (e.g., a large minority may feel that their opinions are not being sufficiently heard).
To overcome this issue, it can be useful to precede the voting phase with a deliberation phase, in which agents may interact, mutually hoping to find some common grounds~\cite{cohen1989deliberation}.
When taken to the extreme, the best outcome of such a deliberation phase is that it would end in consensus:
  i.e., in a situation in which all agents eventually hold the same opinion; when all agents are in consensus. Then, informally speaking, the use of a voting rule is not needed, as all agents would be pleased with choosing the consensus opinion (technically, any \emph{unanimous} voting rule -- that chooses the consensus opinion whenever it exists -- would be accepted).

Naturally, there are many ways by which voting and deliberation may coexist and interact; we discuss some of them in Section~\ref{section:related work}. In this paper, our point of view is that the effect of deliberation is a change of the opinions of the agents (as a simplistic example, a right-wing voter may be more centrist after deliberating with a left-wing voter\footnote{Indeed, the result of such deliberation may be the opposite -- that the right-winger would be radicalized; we do not focus on such cases, but mention them in Section~\ref{section:outlook}.}).
Correspondingly, in our model we view deliberation as a ``black-box'' process whose result is the change of agent opinions. In particular, we do not discuss nor model the specifics of deliberation, but rather model the result of using deliberation in diminishing the opinion distances between agents.

In particular, here we consider an iterative process of deliberation:
  initially, each agent holds to her position, which is modeled as an element of some metric space; then, each iteration consists of an implicit voting step, followed by a discussion step.
  In the voting step, an aggregated outcome is identified using some voting rule; then, in the discussion step, agents slightly change their opinions, to be more inline with the aggregated outcome computed in the voting step (specifically, agents move slightly closer to the aggregated outcome).

Note that the voting rule ingredient of our model only affects the specifics of how agents change their opinions due to discussion, as it affects the aggregated outcome.
Our main interest is to characterize the situations for which such an iterative deliberation process converges, as we view converged processes as successful ones, in particular if the converged configuration is a consensus configuration (i.e., configurations in which all agents share the same opinion).

Specifically, throughout the paper we consider various metric spaces that correspond to certain social choice settings and several voting rules for each of these settings. Then, for each specific realization of our model -- that is, for each metric space $(X, d)$ and voting rule $\mathcal{R}$ -- we analyze whether our iterative deliberation process is guaranteed to be successful (i.e., whether for any initial configuration it is always the case that the agents will end up in consensus).

For the settings that guarantee convergence, we are also interested in worst-case upper bounds for the time needed for such convergence (i.e., for the number of iterations until convergence, in the worst case).
Finally, we are also interested in analyzing the possible results of such iterative deliberation processes, by comparing the initial agent opinions to the consensus opinion reached by such processes, whenever a consensus opinion is reached. 

Indeed, our model is very extreme in assuming that, in each iteration, all agents move slightly closer to the aggregated opinion, in a deterministic way. In Section~\ref{section:outlook} we discuss some relaxations to our model.
Note that the extremeness of our model means that our negative results -- in which we show that an iterative deliberation process need not be successful -- are very strong, as such negative results imply that, for such settings, even a very extremely optimistic deliberation process might not succeed.
Generally speaking, we believe that our results shed more light on the relation between deliberation and voting by effectively distinguish between metric spaces and voting rules that are more problematic wrt.\ deliberation as such for which deliberation has greater potential to be successful.

\paragraph{Paper structure}
After discussing related work (Section~\ref{section:related work}) and formally defining our model (Section~\ref{section:formal model}), we prove general observations that apply to any metric space (Section~\ref{section:general observations}).
Then, we consider deliberation in Euclidean spaces (Section~\ref{section:euclid}), in hypercubes (Section~\ref{section:hypercubes}), and in ordinal elections (Section~\ref{section:ordinal}), and conclude with model relaxations and other avenues for future research (Section~\ref{section:outlook}).
Our results are summarized in Table~\ref{table:results}.

\begin{table}[t]
\caption{Summary of our main results. For each model realization -- i.e., a metric space $(X, d)$ and a voting rule $\mathcal{R}$ -- we report whether convergence is guaranteed, and, if so, what is the upper bound of the number of iterations. VNW (Variable Number of Winners~\cite{faliszewski2020multiwinner}) stands for the set of all subsets of some underlying candidate set, and is modeled via hypercubes; MW (Multi-winner~\cite{mwchapter}) stands for the set of all $k$-size subsets of some underlying candidates set, and is modeled via subsets of hypercubes; and SWF (Social Welfare Functions~\cite{moulin2016handbook}) stands for the set of all rankings over some underlying candidate set.
}
\label{table:results}
{
\begin{center}
	\centering
		\begin{tabular}{c c c | c c c} 
  $X$
& $d$
& $\mathcal{R}$
& Convergence
& Time
& Theorem
\\ \hline
  $\mathbb{R}^T$
& $\ell_1$
& Mean*
& \cmark
& $O\left(\max_{v \in V^0} \frac{d(v, w^0)}{\epsilon}\right)$
& \ref{theorem:mean_l1}
\\
  $\mathbb{R}^T$
& $\ell_2$
& Mean*
& \cmark
& $O\left(max_{v \in V^0} \frac{d(v, w^0)}{\epsilon}\right)$
& \ref{theorem:mean_l2}
\\
$\mathbb{R}^{\geq 3}$
& $\ell_{\infty}$
& Mean*
& \xmark
& \xmark
& \ref{example:bad}
\\
$\mathbb{R}^T$
& $\ell_1$, $\ell_2$
& Median*
& \cmark
& $max_{v \in V^0}\lceil\frac{d(v, w^0)}{\epsilon}\rceil$
& \ref{theorem:median}
\\
$\mathbb{R}^{\geq 3}$
& $\ell_{\infty}$
& Median*
& \xmark
& \xmark
& \ref{example:badtwo}
\\
 & & & &
\\
 VNW
& Hamming
& Majority
& \cmark
& $\max_{v \in V^0}\lceil\frac{d(v, w^0)}{\epsilon}\rceil$
& \ref{theorem:vnw}
\\
VNW
& First changed
& Monotonic
& \cmark
& $\lceil m/\epsilon \rceil$
& \ref{theorem:first changed}
\\
& & & &
\\
MW
& Hamming
& Majority
& \cmark
& $\max_{v \in V^0} \lceil\frac{d(v, w^0)}{\epsilon}\rceil$
& \ref{theorem:mw}
\\
MW
& First changed
& Monotonic
& \cmark
& $\lceil m/\epsilon \rceil$
& \ref{theorem:first changed}
\\
& & & &
\\
SWF
& Arbitrary
& Kemeny
& \cmark
& $max_{v \in V^0} \lceil\frac{d(v, w^0)}{\epsilon}\rceil$
& \ref{theorem:kemeny}
\\
SWF
& Swap
& Monotonic scoring 
& \cmark
& ?
& \ref{theorem:scoring}
\\
SWF
& Swap
& STV
& \cmark
& ?
& \ref{theorem:stv}
\\
SWF
& First changed
& Monotonic
& \cmark
& $\lceil m/\epsilon \rceil$
& \ref{theorem:swf first changed}
\\
\end{tabular}
\end{center}
}
\begin{tablenotes}
    \setlength{\itemindent}{15pt}
    \item[a] * Mean and median both being element-wise.
\end{tablenotes}
\end{table}

\section{Related Work}\label{section:related work}
The most relevant literature pointer to our work is the paper of Bulteau et al.~\cite{bulteau2021aggregation}, in which the authors study aggregation methods for metric spaces. In particular, their model includes a metric space $(X, d)$, where $X$ is the set of elements of the space and $d$ is a metric between pairs of elements of $X$; the opinion of an agent is an element $x \in X$ -- referred to as the agent's \emph{ideal point} -- and the distance~$d$ determines the ordinal preferences of an agent over all $X$, where an agent prefers elements that are closer to its ideal point; a voting rule in their framework is a function that takes $n$ points of $X$ and returns an aggregated point in the metric space as the winner of the election.
The jargon we use in this paper is largely due to Bulteau et al.; viewed from our angle, Bulteau et al.\ study a one-time aggregation process in which voters provide their ideal points and an aggregation method is used to find an aggregated point in the space, while we study an iterative process in which the aggregation method is used iteratively, each time causing the agents to move slightly closer to the aggregated point. Note that Bulteau et al.\ mention that their model may be indeed the basis for studying a process that combines voting and deliberation, as we set to do in the current paper. Indeed, we chose to build on the framework of metric aggregation as it evolves around a notion of distance between opinions (and it is general enough to capture many relevant social choice settings at once); this makes it natural for us to model the effect of deliberation by having each agent change its position to be slightly closer -- according to some distance function $d$ -- to the aggregated point, in each iteration.

There are other works that consider iterative deliberation processes:
  e.g., Fain et al.~\cite{fain2017sequential} consider a process in which, in each iteration, two agents negotiate and move slightly closer to each other's point in the space;
  Elkind et al.~\cite{elkind2021united} consider a process of deliberation in a metric space, concentrating on coalitions that may form around compromise points in the metric space;
  and Garg et al.~\cite{garg2019iterative} consider a model in which all agents are moving in the confined radius of a ball around their compromise point.
There are also works that consider deliberation and aim at capturing the internal mechanics of deliberation~\cite{austen2005deliberation,lizzeri2010sequential,austen2006deliberation}; we, however, similarly to Elkind et al.~\cite{elkind2021united}, abstract away the internal mechanism of deliberation and concentrate on the possibility of reaching consensus by deliberation.

We also mention work on opinion diffusion in social networks~\cite{grandi2017social,faliszewski2018opinion}, in which agents are connected via a social network that affects the opinions of neighbors of agents and thus are propagated throughout the network. Technically, our model can be seen as a model of opinion diffusion where the social network is a complete graph (while in standard opinion diffusion the graph is usually not complete), however we prefer to think about our model as a model of deliberation.
Furthermore, we mention work on iterative voting~\cite{meir2017iterative}, in which agents change their votes iteratively after seeing the current votes of other agents. Technically, our model can be also seen as a model of iterative voting where in each iteration all voters change their vote slightly closer to the current aggregated point (while in standard iterative voting, usually voters strategically change their vote), however, again, we prefer to think about our model as a model of deliberation.

\section{Formal Model}\label{section:formal model}

We describe our formal model, which is parameterized by a metric space $(X, d)$ and a voting rule $\mathcal{R}$. The first three ingredients of our model -- namely, the metric space, the agent population, and the voting rule -- are adapted from the model of Bulteau et al.~\cite{bulteau2021aggregation} -- while the discussion ingredient, which is the center of our work, is novel.

\paragraph{Metric Space}
Let $(X, d)$ be a metric space with $X$ being a set of elements in the metric space and $d : X \times X \to \mathbb{R}$ being a metric function, so that (1) $d$ is \textit{symmetric}, with $d(x, y) = d(y, x)$ for every pair $x, y \in X$, (2) $d$ is \textit{non-negative}, with $d(x, y) \geq 0$ and $d(x, y) = 0\leftrightarrow{x=y}$, and (3) $d$ satisfies the \textit{triangle inequality} i.e., $d(x, z) \leq d(x,y)+d(y,z)$ holds for all $x,y,z\in X$.

\paragraph{Agent Population}
Let $V = \{v_0, \ldots, v_{n-1}\}$ be an agent population. Each agent $v \in V$ is associated with its \emph{initial ideal point}, which is an element $x \in X$, understood as the element of $X$ that is most preferred by $v$. We denote the initial ideal point of agent $v_i$ by $v_i^0$.
(Note that, effectively, the metric space sets the agents' ballot type as well as a distance function between possible ballots.)

\paragraph{Voting rule}
For a metric space $(X, d)$, let $\mathcal{R}$ be a function that takes $n$ elements of $X$ and returns an element $w \in X$. We refer to a set of $n$ elements of $X$ as a \emph{profile} $V$ (of voter ballots) and write $\mathcal{R}(V) = w$ to denote that the \emph{$\mathcal{R}$-winner} of the election (with profile $V$) is $w$.

\paragraph{Deliberation}
We model deliberation as an iterative process, such that, in each iteration, the positions of the voters might change. Initially, the positions of the voters are given by their initial ideal points. Then, in each iteration, we apply the voting rule $\mathcal{R}$; consequently, all voters move slightly closer to the current $\mathcal{R}$-winner:
  specifically, denoting by $v_i^j$ the ideal point of voter $i$ at the beginning of the $j$th iteration (so, in particular, $v_i^0$ are the initial ideal points), and denoting by $V^j=\{v_0^j,\ldots,v_{n-1}^j\}$ and the $\mathcal{R}$-winner of the $j$th iteration by $w^j$ (i.e., $w^j$ is the result of applying $\mathcal{R}$ on $V^j$), we have the following constraints, for some value of $\epsilon$.\footnote{Indeed, for some sparse spaces these two constraints may not be always satisfiable, as agents moving towards the current winner may need to jump ``too far''. In the metric spaces we consider in this paper there is always at least a specific $\epsilon$ for which these constraints are indeed satisfiable.}
  
\begin{constraint}\label{constraint:one}
  $d(v_i^{j + 1}, w^j) = max(0, d(v_i^j, w^j) - \epsilon)\ $.
\end{constraint}

\begin{constraint}\label{constraint:two}
  $d(v_i^{j+1}, v_i^{j})=\epsilon$ unless $d(v_i^{j + 1}, w^j)=0$ and then $d(v_i^{j+1}, v_i^{j}) \leq \epsilon\ $.
\end{constraint}

That is, each voter moves an $\epsilon$-closer to the current winner (unless it is already at most an $\epsilon$-close to the current winner, in which case it moves to the winner itself); the second constraint is to make sure that voters do not ``jump around'' too arbitrarily.

We say that the iterative deliberation process \emph{converges} if all agents cease to move after some finite number of iterations; note that, when an iterative deliberation process converges all agents are in consensus.

We say that convergence is guaranteed for some metric space $(X, d)$ and voting rule $\mathcal{R}$ if all possible deliberation processes \emph{converge}, for every $\epsilon > 0$; note that in each iteration, the agents can have multiple options to move, sometimes even an infinite number of options, and there may be an infinite number of initial profiles, so there may be an infinite amount of different deliberation processes.

\begin{example}\label{example:one}
Let $(X, d)$ be with $X$ being $\mathbb{Z}$ and $d(x, y) = |x - y|$. Let $\mathcal{R}(V) = \sum_{v \in V} \lfloor v_i / n \rfloor$ (so the $\mathcal{R}$-winner is the average, rounded down). Let $V = \{v_0, v_1, v_2\}$ with $v_0^0 = 3$, $v_1^0 = 5$, and $v_2^0 = 8$. Let $\epsilon$ be $1$. Then, the iterative deliberation process proceeds as follows:
  (1) at the beginning of the first iteration, $v_0$ stands on $3$, $v_1$ on $5$, and $v_2$ on $8$. The $\mathcal{R}$ winner is $w := \lfloor (3 + 5 + 8) / 3 \rfloor = 5$. Now, each $v_i$ moves an $\epsilon$-closer to~$5$; 
  (2) at the beginning of the second iteration, $v_0$ stands on $4$, $v_1$ remains on $5$, and $v_2$ stands on $7$. The $\mathcal{R}$ winner is again $w := 5$;
  (3) at the beginning of the third iteration, $v_0$ stands on $5$, $v_1$ remains on $5$, and $v_2$ stands on $6$. The $\mathcal{R}$ winner is again $w := 5$;
  (4) at the beginning of the fourth iteration, $v_0$ stands on $5$, $v_1$ remains on $5$, and $v_2$ stands on $5$. The $\mathcal{R}$ winner is again $w := 5$.
In particular, for this example, the iterative deliberation process converges, as, after the fourth iteration, nobody would move. See Figure~\ref{figure:one}.
\end{example}

\begin{figure}[t]
  \begin{minipage}[c]{0.67\textwidth}
\includegraphics[width=7.5cm]{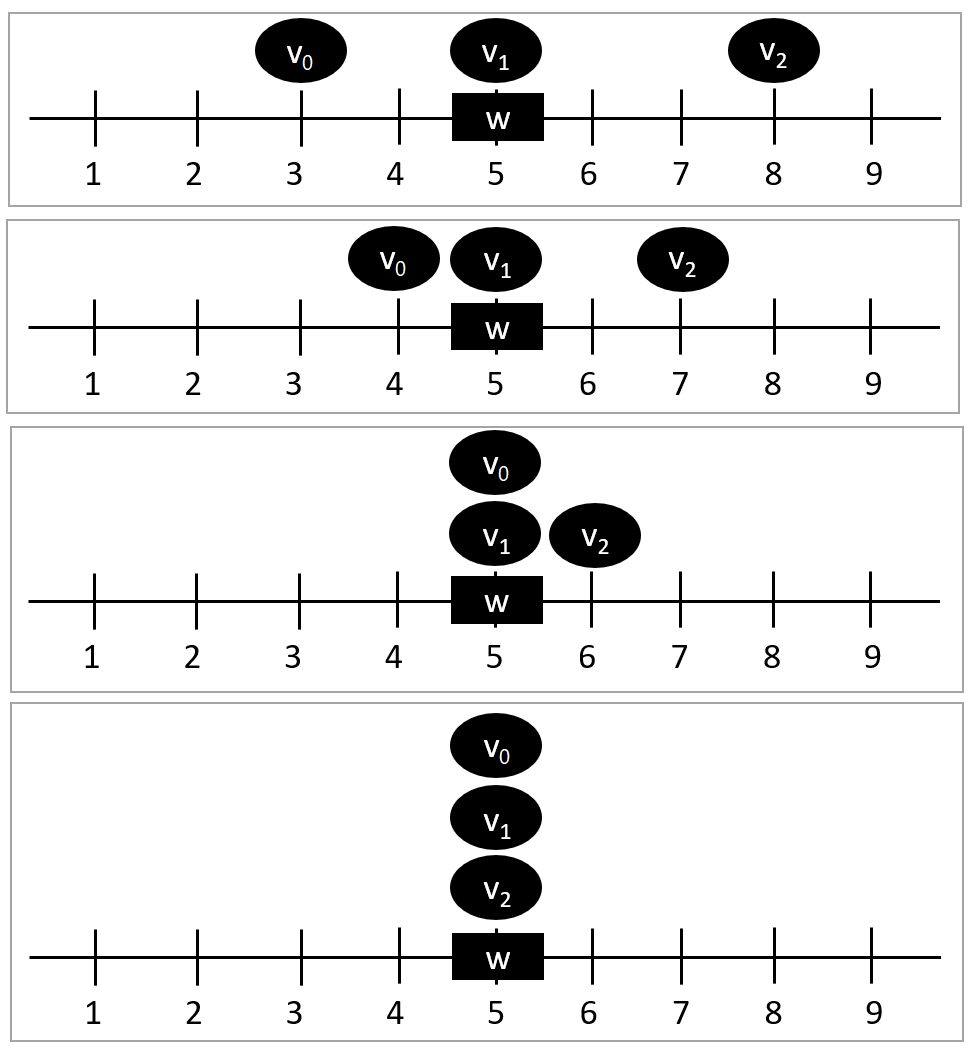}
  \end{minipage}\hfill
  \begin{minipage}[c]{0.3\textwidth}
    \caption{Illustration for Example~\ref{example:one}. The top box shows the initial configuration, with voter $v_0$ at $3$, voter $v_1$ at $5$, and voter $v_2$ at $8$, implying that the aggregated point $w$ is at $5$. The second box from the top depicts the situation at the second iteration, each box below shows the situation after another iteration, and the process converges at the fourth iteration, as shown in the bottom box, in which all voters are in consensus at $5$.} \label{figure:one}
  \end{minipage}
\end{figure}

\section{General Observations}\label{section:general observations}

We begin with some observations, regarding general sufficient conditions that guarantee convergence (later we will discuss specific metric spaces).

First, we observe that, whenever the process converges, it indeed converges to consensus.
This follows as, if the voters are not in consensus, then at least one voter would move in the current iteration.

\begin{observation}
  In a converged configuration, the profile is consensus.
\end{observation}

The next theorem roughly says that, if the winner does not move too much, then deliberation is guaranteed; the proof follows the intuition that, if all agents move $\epsilon$-closer to the winner, but the winner moves ``slower'' than this (in particular, by less than $\epsilon$ by a $\delta>0$, for if the demand was only less than $\epsilon$, we could have just approached convergence in the limit where the number of iterations approaches $\infty$), then eventually the process shall converge.
\begin{theorem}\label{theorem:one}
  Consider some $(X, d)$ and $\mathcal{R}$.
  If, for each profile $V$ and for each~$\epsilon$, there is an index $k$ such that for each $j \geq k$ there exists $\delta > 0$, and it holds that $d(w^j, w^{j+1}) \leq \epsilon - \delta$,
  then convergence is guaranteed.
\end{theorem}

\begin{proof}
The following holds, by the triangle inequality:
  $\sum_i(d(v_i^{j+1},w^{j+1}))\leq \sum_i(d(v_i^{j+1},w^j) + d(w^j,w^{j+1}))$.
Using Constraint~\ref{constraint:one}, the theorem's assumption, and contraction of the $\epsilon$, it follows that $\sum_i(d(v_i^{j+1},w^{j+1}))\leq\sum_i(d(v_i^j,w^j))-\delta$.
Thus, the sum of all distances of all $v_i$ from $w$ is decreasing, by at least a $\delta$ in each iteration, from index $k$, until it reaches zero, in a maximum of $\sum_i(d(v_i^0,w^0))/\delta$ iterations, and convergence follows.
\qed\end{proof}

The next theorem deals with the time complexity of the iterative process and the winner of the last iteration.

\begin{theorem}\label{theorem:time}
  For any $(X, d)$ and $\mathcal{R}$ where the iterative process is such that for every profile $V$, the $\mathcal{R}$-winner of the profile reached in the next iteration $V'$ is equal to the $\mathcal{R}$-winner of $V$, and let $D=max_{v \in V_0} d(v, w^0)$, the maximal distance of any agent from the $\mathcal{R}$-winner of the initial state, $w^0$, then the number of iterations until convergence is reached is exactly $\lceil D/\epsilon \rceil$ and the $\mathcal{R}$-winner of the last iteration is $w^0$.
\end{theorem}

\begin{proof}
  From Theorem~\ref{theorem:one} convergence is guaranteed, and since for any profile the $\mathcal{R}$-winner doesn't change, so the $\mathcal{R}$-winner in the last iteration must be $w_0$. Also, the voter which is the farthest from $w$, must reduce its distance in each iteration by $\epsilon$, so in each iteration the maximal distance reduces by $\epsilon$, so in exactly $D/\epsilon$ iterations, convergence shall be reached.
\qed\end{proof}

The next theorem roughly says that, if the agents move to the center of mass, then convergence is guaranteed.

\begin{theorem}\label{theorem:mass}
  For any $(X, d)$,
  if $\mathcal{R} = \argmin_{x \in X} \sum_{v \in V} d(v, x)$, then convergence is guaranteed.
\end{theorem}

\begin{proof}
We show that in all iterations, the winner stays the same, and therefore, by Theorem~\ref{theorem:one}, convergence follows.
Let $w$ denote the winner for some iteration $V$, so $w$ has the minimum sum of distances from the agents in that iteration - $w=\argmin_{x \in X} \sum_{v \in V} d(v, x)$. In the next iteration $V'$, suppose that some element $y$ is the winner - $y=\argmin_{x \in X} \sum_{v \in V'} d(v, x)$
Note that every agent got closer to $w$ by $\epsilon$, unless it was already less than $\epsilon$-far away from $w$, in which case it moved to $w$. We denote by $V_1$
the group of agents that were $\epsilon$ or more far away from $w$, and $V_2=V-V_1$. So, the new sum of distances from $w$ is $d_w=\sum_{v \in V'} d(v, w) = \sum_{v \in V_1}(d(v, w) - \epsilon) + \sum_{v \in V_2} d(v, w) - \sum_{v \in V_2} (v, w)$ where we replaced $V_1'$ by $V_1$ by reducing $\epsilon$, and just added and subtracted all elements in $V_2$. $d_w = \sum_{v in V} d(v, w) - n\epsilon - \sum_{v \in V_2} d(v, w)$, where we joined $V_1$ and $V_2$ into $V$.
$d_x=\sum_{v \in V'} d(v, x)$. But by the triangle inequality, for $v' \in V'$ and the corresponding $v \in V$, $d(x, v') + d(v, v') \geq d(x, v)$, and by Constraint~\ref{constraint:two} $d(v, v') \leq \epsilon$, so $d(x, v') \geq d(x, v) - \epsilon$.
Replacing that in the first equation, $d_x \geq \sum_{v \in V} (d(v, x) - \epsilon)=\sum_{v \in V} d(v, x) - n\epsilon$.
If we combine the two values, we get $d_x\geq d_w$, because the first sum is just the sum of distances in iteration $V$, where $w$ was the argmin. Thus, $x = w$ and the claim follows by Theorem~\ref{theorem:one}.
\qed\end{proof}

\section{Deliberation in Euclidean Spaces}\label{section:euclid}

In this section we consider Euclidean spaces; these are natural spaces that are studied extensively in social choice~\cite{arrow1990advances}.
Formally, we consider metric spaces $(X, d)$ in which $X$ is a $T$-dimensional\footnote{We use ``$T$'' and not the standard ``$d$'', as ``$d$'' is taken by the metric space $(X, d)$.} Euclidean space, $X = \mathbb{R}^T$, for some $T = \{1, 2, 3, \ldots\}$; as for the distance function $d$, we consider $\ell_p$ norms for $p \in \{1, 2, \infty\}$ (other distances are indeed possible, however here we concentrate on $\ell_p$ norms).
As for the voting rule $\mathcal{R}$, we consider the element-wise mean and median, as formally defined below.

\begin{definition}
    The \emph{element-wise mean} of $n$ points/voters $v_i$, $i \in [n]$, in some~$\mathbb{R}^T$, is a point in $\mathbb{R}^T$ such that the value in each dimension is the mean of the values of the voters in that dimension; that is, the value at dimension $t$, $t \in [T]$, is $\sum_{i \in [n]} v_i[t] /n$.
\end{definition}

\begin{definition}
    The \emph{element-wise median} of $n$ points/voters $v_i$, $i \in [n]$, in some~$\mathbb{R}^T$ is a point in $\mathbb{R}^T$ such that the value in each dimension is the median of the values of the voters in that dimension; that is, the value at dimension $t$, $t \in [T]$, is $median(v_i[t])$, where we define the median of an even number of real numbers to be the larger between the two middle numbers.
\end{definition}

First, we show that, for $\ell_1$, every coordinate of every agent moves closer to the winner and does not pass it.

\begin{lemma}\label{lemma:l1}
  Let $(X, d)$ be such that $X = \mathbb{R}^T$ and $d$ is $\ell_1$. Then for any agent $v$, element $i$ and iteration $j$, it holds that either $w[i] \leq v^{j+1}[i] \leq v^{j}[i]$ or $w[i] \geq v^{j+1}[i] \geq v^{j}[i]$ where $w$ is the winner of the $j$'s iteration.
\end{lemma}

\begin{proof}
  In $\ell_1$ the contribution of every coordinate is just added with absolute; thus, if $v^{j+1}[i]$ is not between $v^{j}[i]$ and $w[i]$, then the contribution of $v^{j+1}[i]$ to $d(v^{j+1}, w)$, will need to be compensated by the other coordinates to accommodate for the $\epsilon$ reduction, and then $d(v^{j}, v^{j+1})$ will be greater than $\epsilon$.
\qed\end{proof}

We use Lemma~\ref{lemma:l1} to show that convergence is guaranteed, for $\ell_1$, for element-wise mean.

\begin{theorem}\label{theorem:mean_l1}
  Let $(X, d)$ be such that $X = \mathbb{R}^T$ and $d$ is $\ell_1$, and let $\mathcal{R}$ be the element-wise mean. Then, convergence is guaranteed.
\end{theorem}

\begin{proof}
  We know that each coordinate can only approach the mean. So each voter has an $\epsilon$ total distance to move to change the mean, because it is $\ell_1$. But in each dimension, some voters are larger than the mean, and some voters are smaller than it, so they will contribute contradictory values, and thus the mean will move strictly less than $\epsilon$, and by Theorem~\ref{theorem:one} convergence is guaranteed.
\qed\end{proof}

Next we show that, in the case of $\ell_2$, each voter, under our constraints, must move directly to the $\mathcal{R}$-winner in a straight line.

\begin{lemma}\label{lemma:l2}
  Let $(X, d)$ be such that $X = \mathbb{R}^T$ and $d$ is $\ell_2$, then for every agent $v$ and iteration $j$, there exists only one point $v^{j+1}$ that adheres to the constraints, and that point is an $\epsilon$ closer on a straight line from $v^{j}$ to the $\mathcal{R}$-winner of the $j$'s iteration.
\end{lemma}

\begin{proof}
  Let us assume, wlog., that $w=0$, and $v^{j}[i']=0$ for every $i'$ except $i$. We can do this by displacement and rotation to the $i$'s axis, both are $\ell_2$-invariant.
  We define $M=d(v^{j}, w)$. So we have both vectors sizes $v^{j}[i]=M$, $\sum_{i'}v^{j+1}[i']^2=(M-\epsilon)^2$, The second size being because of Constraint~\ref{constraint:one}.
  $d(v^{j},v^{j+1})^2=\sum_{i'}(v^{j}[i']-v^{j+1}[i'])^2=\sum_{i'}v^{j}[i']^2+\sum_{i'}v^{j+1}[i']^2-2\sum_{i'}v^{j}[i']\cdot v^{j+1}[i']=M^2-(M-\epsilon)^2-2v^{j}[i]\cdot v^{j+1}[i]=2M^2+\epsilon^2-2M\cdot\epsilon-2v^{j}[i]\cdot v^{j+1}[i]=\epsilon^2$.

  Where we used our assumptions, vector sizes, and the fact that the distance between $v^{j}$ and $v^{j+1}$ is $\epsilon$.
  We can write $v^{j}[i]\cdot v^{j+1}[i]=M^2-M\cdot\epsilon=M\cdot(M-\epsilon)=v^{j}[i]\cdot\sqrt{\sum_{i'} v^{j+1}[i']^2 }$. Or, $\sum_{i'} v^{j+1}[i']^2=v^{j+1}[i]^2$. Or, if we just subtract $v^{j+1}[i]^2$, then all coordinates except $i$ in $v^{j+1}$ are zero. Now, $d(v^{j+1},w)=d(v^{j},w)-\epsilon$. So, $v^{j+1}[i]=v^{j}[i]-\epsilon$.
\qed\end{proof}

We use Lemma~\ref{lemma:l2} to show that, with $\ell_2$ and the element wise mean, convergence is guaranteed.

\begin{theorem}\label{theorem:mean_l2}
  Let $(X, d)$ be such that $X = \mathbb{R}^T$ and $d$ is $\ell_2$, and let $\mathcal{R}$ be the element-wise mean. Then, convergence is guaranteed.
\end{theorem}

\begin{proof}
  If we look in a coordinate system in which only one axis aligns with the vector between each $v$ and the mean (rotation is invariant to $\ell_2$), then we can see that the contribution of $v$ is just $\epsilon / n$ in the direction from $v$ to the mean, by Lemma~\ref{lemma:l2}, because only one axis has a delta which is not 0. Because it cannot be that all those vectors are in the same direction (in that case, the mean would be closer in the opposite direction until it passed one of them), it follows that the mean moves strictly less than $\epsilon$; by Theorem~\ref{theorem:one}, convergence follows.
\qed\end{proof}

Next, we look at the time order and the last converged winner in $\ell_1$ and $\ell_2$ in element wise-mean, and we show that if we look at the smallest T-dimensional ball that contains all the agents, then the time order is of the diameter of that ball divided by $\epsilon$, and the converged winner is inside that ball.

\begin{theorem}\label{theorem:killone}
  Let $(X, d)$ be such that $X = \mathbb{R}^T$ and $d$ is $\ell_1$ or $\ell_2$, and let $\mathcal{R}$ be the element-wise mean, and let $D$ be the diameter of the smallest $T$-dimensional ball that contains all $n$ agents. Then, the process would converge to a point inside that ball in $O(D/\epsilon)$ iterations.
\end{theorem}

\begin{proof}
In every iteration, the mean cannot be on the circumference of the ball, as at least one agent has to be on either side of the ball, otherwise the ball would be smaller (so the mean has a contribution from the agent from the other side, that keeps it away from the circumference of the agent on either side).
Furthermore, as each agent reduces its distance to the mean by~$\epsilon$, and as the proofs of Theorems~\ref{theorem:mean_l1} and~\ref{theorem:mean_l2} show that the mean moves strictly less than $\epsilon$ in both cases, it means that the diameter of the ball must reduce by $\Omega(\epsilon)$ in each iteration. And so, in $O(D/\epsilon)$ iterations, the diameter would reach zero, and convergence would be reached. Also, because the mean is inside the ball in each iteration, and the agents move closer to the mean, by Lemma~\ref{lemma:l1} and Lemma~\ref{lemma:l2}, the next mean would stay in the ball, implying that also the final mean would be inside the ball.
\qed\end{proof}

Now we show, using both Lemma~\ref{lemma:l1} and Lemma~\ref{lemma:l2}, that for both $\ell_1$ and $\ell_2$, in the case of element-wise median, convergence is guaranteed.

\begin{theorem}\label{theorem:median}
  Let $(X, d)$ be such that $X = \mathbb{R}^T$ and $d$ is $\ell_1$ or $\ell_2$, and let $\mathcal{R}$ be the element-wise median. Then, convergence is guaranteed.
\end{theorem}

\begin{proof}
  In $\ell_1$, by Lemma~\ref{lemma:l1}, each coordinate stays between the $\mathcal{R}$-winner and the last coordinate, and in $\ell_2$, by Lemma~\ref{lemma:l2}, also each coordinate stays between the $\mathcal{R}$-winner and the last coordinate, so the median in each coordinate does not change, and because $d$ is the element-wise median, the $\mathcal{R}$-winner does not change, and by Theorem~\ref{theorem:one}, convergence is guaranteed.
\qed\end{proof}

The time order and final winner are known from Theorem~\ref{theorem:time}.

\begin{corollary}
  Let $(X, d)$ be such that $X = \mathbb{R}^T$ and $d$ is $\ell_1$ or $\ell_2$, and let $\mathcal{R}$ be the element-wise median, and let $D=max_{v \in V^0} d(v, w^0)$, then the number of iterations until convergence is reached is exactly $\lceil D/\epsilon \rceil$ and the $\mathcal{R}$-winner of the last iteration is $w^0$.
\end{corollary}

In contrast to the results above, below we show that, for $\ell_{\infty}$, for the element-wise mean and element wise median, convergence is not guaranteed, at least for $T$-dimensional Euclidean spaces with $T \geq 3$. We show this by two counter examples.

The first example deals with element-wise mean.

\begin{example}\label{example:bad}
Set $\epsilon=1$, $n=3$, 
$v^0_0=(-4,2,2)$,
$v^0_1=(2,-4,2)$,
$v^0_2=(2,2,-4)$.
Running the iterative process for these initial conditions would result in adding~$1$ to each of the voters, in each dimension, in each step of the process.
I.e.:
  $v_0^1=(-3,3,3)$,
  $v_1^1=(3,-3,3)$,
  $v_2^1=(3,3,-3)$;
and generally:
  $v_0^j=(-4+j,2+j,2+j)$, $v_1^j=(2+j,-4+j,2+j)$, $v_2^j=(2+j,2+j,-4+j)$.
Indeed, this is an endless behavior, exploding to infinity.
It is possible to adapt both this example and the following one to any dimension $T > 3$, by adding as many zero dimensions as needed.
\end{example}

The second example deals with element-wise median.

\begin{example}\label{example:badtwo}
  The same pattern repeats for the element wise median
  $v_0^0=(0,0,0)$, $v_1^0=(-2,0,0)$, $v_2^0=(0,-2,0)$, $v_3^0=(0,0,-2)$.
  $v_0^{j+1}=(j,j,j)$, $v_1^{j+1}=(-1+j,1+j,1+j)$, $v_2^{j+1}=(1+j,-1+j,1+j)$, $v_3^{j+1}=(1+j,1+j,-1+j)$. As we defined the median to be the larger number when there is an even number of agents, the median increases by $1$ in every dimension in every iteration.
\end{example}

\begin{remark}
Generally speaking, non-convergence can be due to two possibilities:
  (1) Getting stuck in a cycle;
  or
  (2) moving to infinity.
Note that Example~\ref{example:bad} and Example~\ref{example:badtwo} are of the second type, which, in a way, is more dramatic; and, perhaps, less intuitive.
\end{remark}

\section{Deliberation in Hypercubes}\label{section:hypercubes}

Next we consider $T$-dimensional hypercubes; these spaces naturally correspond to multiple referenda~\cite{ma-comsoc} as well as to multiwinner elections~\cite{mwchapter} and committee selection with variable number of winners~\cite{faliszewski2020multiwinner}; below, we consider the latter two settings separately (indeed, for convenience, we use the jargon of multiwinner elections). We consider approval ballots here (in the next section we consider ordinal elections).

An important point to make, to all the discrete metric spaces we look at, is that we consider only $\epsilon \in \mathcal{N}$, because otherwise our metric spaces would be too sparse, and our constraints could not be met.

\subsection{Committee Elections with Variable Number of Winners}

The social choice setting here consists of a set of candidates and a set of agents such that each agent provides a subset of the candidates; then, a subset of the candidates -- without restrictions on its size -- is to be selected as the winner of the election.
This setting is studied under the umbrella of committee elections with variable number of winners~\cite{faliszewski2020multiwinner}; following the literature, we refer to this setting as \emph{VNW}.

Formally, we have a metric space $(X, d)$, where $X = \{0, 1\}^m$ for some integer~$m$ and $d(u, v)$ is the Hamming distance.
An important class of VNW rules are monotonic rules, as defined next.

\begin{definition}
  A VNW rule $\mathcal{R}$ is \emph{monotonic} if the following holds:
    for each profile $V$ and its $\mathcal{R}$-winner $w$, it holds that the $\mathcal{R}$-winner $w'$ for $V'$, where $V'$ is similar to $V$ except for one agent that either (1) flips some $1$ to $0$ for some candidate not in $w$; or (2) flips some $0$ to $1$ for some candidate in $w$; then $w' = w$ (i.e., $w$ stays).
\end{definition}

We show a rather general result next, applying to all monotonic VNW rules.
Indeed, many VNW rules are monotonic:
  in particular, Majority is.

\begin{theorem}\label{theorem:vnw}
  Let $(X, d)$ be such that $X$ is VNW and $d$ is the Hamming distance, and let $\mathcal{R}$ be a monotonic VNW rule, then convergence is guaranteed.
\end{theorem}
\begin{proof}
In each iteration, each agent $v$ must reduce its Hamming distance from $w$ by $\epsilon$. So, that means it must flip $\epsilon$ bits that are either $1$ in $v$ and $0$ in $w$, or $0$ in $v$ and $1$ in $w$. (If $d(v,w)<\epsilon$ then it flips fewer bits, or even 0 bits if it coincides with $w$.)
Now, using -- for $m\cdot\epsilon$ times -- the fact that $\mathcal{R}$ is \emph{monotonic}, we deduce that the winner stays the same; the result then follows from Theorem~\ref{theorem:one}.
\qed\end{proof}

The time order and final winner are known from Theorem~\ref{theorem:time}.

\begin{corollary}
  Let $(X, d)$ be such that $X$ is VNW and $d$ is the Hamming distance, and let $\mathcal{R}$ be a monotonic VNW rule, then $\mathcal{R}$-winner of the last iteration is the $\mathcal{R}$-winner of the first iteration, $w^0$, and the number of iterations is $\max_{v \in V^0} \lceil d(v, w^0)/\epsilon\rceil$.
\end{corollary}

\subsection{Committee Elections}

Here we consider the standard model of multiwinner elections, in which we have a set of candidates, a set of agents such that each agent provides a subset of the candidates; then, a subset of the candidates -- of some predefined size $k$ (note the difference to the section above, in which such a $k$ was not given) -- is to be selected as the winner of the election~\cite{mwchapter}; following the literature, we refer to this setting as \emph{MW}.

It turns out that, even though the setting of MW is different from that of VNW -- technically, only a subset of the hypercube is admissible (in particular, only those elements of the metric space that correspond to $k$ candidates) -- in the context of the deliberation process, our general result regarding monotonic rules is similar to that for VNW.
Correspondingly, the proof of the result below follows similar lines as the proof of Theorem~\ref{theorem:vnw}.

\begin{theorem}\label{theorem:mw}
  Let $(X, d)$ be such that $X$ is MW and $d$ is the Hamming distance, and let $\mathcal{R}$ be a monotonic MW rule, then convergence is guaranteed.
\end{theorem}

The time order and final winner are known from Theorem~\ref{theorem:time}.

\begin{corollary}
  Let $(X, d)$ be such that $X$ is MW and $d$ is the Hamming distance, and let $\mathcal{R}$ be a monotonic MW rule, then the $\mathcal{R}$-winner of the last iteration is the $\mathcal{R}$-winner of the first iteration, $w_0$, and the number of iterations is exactly $\max_{v \in V^0} \lceil d(v, w^0)/\epsilon\rceil$.
\end{corollary}

\subsection{The First Changed Distance}

We consider another metric distance, that measures the distance by the index of the first candidate that is different, which imitates the $\ell_\infty$ property in Euclidean spaces, such that it can ignore some changes in other candidates as long as one relevant candidate changes its place in the right direction. We tried to imitate $\ell_\infty$ in order to get a result of non-convergence, but unfortunately this distance converged with every monotonic voting rule.

\begin{definition}
  The \emph{first changed} distance is defined as a function $d:X \times X \to N$ which is equal to $d(v_1,v_2) = \argmin_i{v_1[i]\neq v_2[i]}$.
\end{definition}

\begin{theorem}\label{theorem:first changed}
  Let $(X, d)$ be such that $X$ is MW or VNW and $d$ is the \emph{first changed} distance, and let $\mathcal{R}$ be a monotonic voting rule, then convergence is guaranteed, and the number of iterations is $\lceil m/\epsilon \rceil$.
\end{theorem}

\begin{proof}
  In the first iteration we look at the last $\epsilon$ candidates of the agents and $w$. All agents that have the same last $\epsilon$ candidates as $w$ have a distance equal to or less than $m-\epsilon$ and it will reduce by $\epsilon$ so it will stay that way. All agents that have different candidates than $w$ in the last $\epsilon$ places must reduce their distance by $\epsilon$, so they must change the last $\epsilon$ candidates to that of the winner. Since the voting rule is monotonic, the last $\epsilon$ candidates of the winner stay the same. So now, all agents have a distance of $m-\epsilon$ from the new $w$. We repeat this logic for $\lceil m/\epsilon \rceil$ iterations, until convergence is reached.
\qed\end{proof}

\section{Ordinal Elections}\label{section:ordinal}

Here we consider the standard ordinal model of elections~\cite{moulin2016handbook}:
  in this setting there is a set of candidates and a set of agents such that each agent provides a linear order (i.e., a ranking, or, equivalently, a permutation) over the set of candidates; then, the result of the aggregation method -- that is usually called a \emph{social welfare function} -- is an aggregated ranking; following the literature, we refer to this setting as \emph{SWF}.

Formally, we have a metric space $(X, d)$ where $X$ is the set of linear orders over some underlying set of candidates and $d$ is the swap distance (of course, other distances are possible, however the swap distance is perhaps the most natural and most popular distance in this context~\cite{hogrebe2019complexity,faliszewski2019similar,faliszewski2020isomorphic}). 

In this setting too, we consider only $\epsilon \in \mathcal{N}$, as explained in Section~\ref{section:hypercubes}.

As for the voting rule $\mathcal{R}$, first we observe that, 
as Kemeny is the realization of $\argmin_{x \in X} \sum_{v \in V} d(v, x)$ for this context, the next result follows Theorem~\ref{theorem:mass}.

\begin{corollary}\label{theorem:kemeny}
  Let $(X, d)$ be such that $X$ is SWF and $d$ is any distance, and let $\mathcal{R}$ be Kemeny, then convergence is guaranteed.
\end{corollary}

And from Theorem~\ref{theorem:time}, the winner of the last iteration is the winner of the first iteration, and the number of iterations is $max_{v \in V_0} d(v, w_0)/\epsilon$.

\begin{corollary}
  Let $(X, d)$ be such that $X$ is SWF and $d$ is any distance, and let $\mathcal{R}$ be Kemeny, then the $\mathcal{R}$-winner of the final iteration is $w^0$, and the number of iterations is $max_{v \in V^0} \lceil d(v, w^0)/\epsilon \rceil$.
\end{corollary}

As for other voting rules, we provide a rather general result, following the next definition.

\begin{definition}
  An SWF rule $\mathcal{R}$ is a \emph{scoring rule} if it corresponds to a function $f : X^n \to \mathbb{R}^m$ (i.e., it takes a profile of n agents, and assigns an individual score (real number) to each candidate - m in the number of candidates), such that it chooses the $\mathcal{R}$-winner by sorting the candidates in decreasing order of their scores (Ties can be handled by an arbitrary, fixed order $O$ over the candidates).
\end{definition}

\begin{definition}
  A \emph{scoring rule} is a \emph{monotonic scoring rule} if for every two profiles $V$ and $V'$, if a candidate $c$ is ranked at least as high in $V$ compared to $V'$ for every agent, then $f(V)[c] \geq f(V')[c]$. 
\end{definition}

\begin{theorem}\label{theorem:scoring} 
  Let $(X, d)$ be such that $X$ is SWF and $d$ is the swap distance, and let $\mathcal{R}$ be a \emph{monotonic scoring rule}, then convergence is guaranteed.
\end{theorem}

\begin{proof}
The proof follows a potential function argument. To this end, we define a potential function that assigns a vector to each profile, and we define a lexicographic order on these vectors, and show that each iteration of the deliberation process can only advance in that order in one direction.
We also show that the only way that we can stop advancing is if we are in consensus, in which case we have reached a maximum and the process would halt.

More formally, for a profile $V$, denote by $w$ the $\mathcal{R}$-winner of $V$. Then, define a triplet for each candidate with index $i$ in $w$ (denoted by $w_i$), as follows: $(f(V)[w_i], O[w_i], B(V)[w_i])$, where $B(V)[w_i]$ is the Borda score of candidate $w_i$. (i.e., for each candidate $c$, $B(V)[c] := \sum_{v \in V} m - pos_v(c)$, where $pos_v(c)$ is the position of $c$ in the vote of $v$.) Then, define a vector combining all triplets in the order that their respective candidates appear in $w$, and consider an order on all profiles $V$ according to the lexicographic order of these vectors.

\begin{example}\label{example:intheproof}
Let $\mathcal{R}$ be Plurality, let the set of candidates be $\{a, b, c\}$, let $O = (a,b,c)$, and let $V = \{v_0, v_1, v_2\}$ with $v_0 = \{a, b, c\}$, $v_1 = \{a, b, c\}$ and $v_2 = \{c, a, b\}$. Then $w = \{a, c, b\}$, and the triplet for $a$ is $(2, 2, 5)$, for $c$ it is $(1, 0, 2)$, and for $b$ it is $(0, 1, 2)$. The combined vector for $V$ is thus $(2, 2, 5, 1, 0, 2, 0, 1, 2)$.
\end{example}

In case of consensus, no swap is made, thus the profile remains the same, so we stay with the same vector and place in the order.
Otherwise, we look at the index $i$ in the $\mathcal{R}$-winner (of the iteration before the swap) of the first candidate that was swapped in some agent. 
First we notice, that it could not have been swapped backwards in any of the agents, because that would violate constraint~\ref{constraint:one} (that the distance from $w$ must reduce by $\epsilon$; we could not have swapped it back, and make up for it by swapping forward one more candidate, because that would violate constraint~\ref{constraint:two}).
Next, its potential must have increased, because at the least, its Borda score must have increased (Borda is strictly monotonic); and because our scoring rule is monotonic, its scoring function did not decrease, and $O$ stayed the same because it is constant. Also, all the candidates in front of candidate indexed $i$ in $w$ did not swap, so their scoring function, Borda score and $O$ did not decrease.

Thus, so far we showed that if all the candidates in the new $\mathcal{R}$-winner until place $i$ remained in the same order as the old one, our potential must have increased, and our proof is done. If their order has changed, then we look at the first index $j \leq i$ that has changed (that replaced its candidate). Now, because the switch had occurred, we know that the new candidate must have a higher $f$ score, or the same $f$ score, and a higher $O$ score, by our definition of the $\emph{monotonic scoring rule}$, and so, our potential increased. So, we proved that the potential must increase in all cases, and so we must advance in our order, until consensus is reached.
\qed\end{proof}

As Plurality, Borda, and Copeland are all monotonic scoring rules, they all converge.

\begin{corollary}
  Let $(X, d)$ be such that $X$ is SWF and $d$ is the swap distance, with $\mathcal{R} \in \{$Plurality, Borda, Copeland$\}$, then convergence is guaranteed.
\end{corollary}

We look at another voting rule, STV, with swap distance.

\begin{definition}
  STV is a SWF rule that chooses the winner as follows:
  $V_0$ is set to be $V$, the input profile of the voting rule;
  then, in iteration $k$ (of the STV rule procedure, not the iterative process), the Plurality looser of $V_k$, $c$, is determined, and $w[m-1-k]$ is set to  be $c$. Then $V_{k+1}$ is set to be $V_k$, with all instances of $c$ removed, leaving $m-1-k$ candidates in each agent.
  This process repeats for m iterations, until $w$ is filled, from the last place to the first.
  Note that because the definition uses Plurality, then there is also a vector $O$ that defines an order on the candidates for it.
\end{definition}

\begin{theorem}\label{theorem:stv}
  Let $(X, d)$ be such that $X$ is SWF and $d$ is the swap distance, and let $\mathcal{R}$ be STV, then convergence is guaranteed.
\end{theorem}

\begin{proof}
  The same kind of proof as with Theorem~\ref{theorem:scoring} is applied here, with a few changes: we decrease in the potential value of the vector in each iteration (instead of increase); we order the vector from the last place in $w$ to the first; lastly, the first element of the triplet (denoted the STV score) is defined as the Plurality score of the candidate in the STV iteration (of the STV rule procedure) that it was removed from $V$, with the Plurality score defined as the number of agents that the candidate is in their first place in $V$.
  
  If we look at the same example from Theorem~\ref{theorem:scoring} (Example~\ref{example:intheproof}), recalling that $O = (a,b,c)$, $v_0 = \{a, b, c\}$, $v_1 = \{a, b, c\}$ and $v_2 = \{c, a, b\}$; then, in the first STV iteration, $b$ is the Plurality looser and is removed from $V$, and we are left with $v_0 = \{a, c\}$, $v_1 = \{a, c\}$ and $v_2 = \{c, a\}$. In the second STV iteration $c$ is the Plurality looser and is removed from $V$, and then we are left with only $a$, so it is in the first place of $w$. So $w=\{a, c, b\}$, and the STV score for $a$ is 3, for $c$ it is 1, and for $b$ it is 0. And the combined vector is thus $(0, 1, 2, 1, 0, 2, 3, 2, 5)$.
  
  To continue with the proof, we look at the last candidate in $w$, denoted as~$c$. It was removed from $V$ in the first STV iteration (that is the reason that it is the last place in $w$). The same logic as in the proof of Theorem~\ref{theorem:scoring} shows that $c$ could not have been swapped forward in any of the agents (it would have violated our constraints otherwise). So it could only have been swapped backwards, thus only reducing its Plurality score in the first STV iteration. If it still in the last place in the next $\mathcal{R}$-winner, then it was still the Plurality looser in the first STV iteration, and either it has been swapped backwards in some agent, and its Borda score reduced, and we are done (because its Plurality score and $O$ have not increased), or it hasn't been swapped in any of the agents. In case some other candidate, $c'$, took its place in $w$, we know that either the Plurality score of $c'$ was lower than that of $c$, and that is why it was removed in the current STV iteration, or they have the same STV score, and $c'$ is lower than $c$ in the predefined $O$ order. In both cases, we decreased in the potential, and we are done also.
  If $c$ has not been swapped in any of the agents and stayed in the last place, we can just remove all instances of $c$ from $V$, and repeat the same logic, until we reach some candidate that has been swapped in some agent, or none has, and we are in consensus.
\qed\end{proof}

The same proof as with committee elections with \emph{first changed} distance works in Ordinal elections as well.

\begin{corollary}\label{theorem:swf first changed}
  Let $(X, d)$ be such that $X$ is SWF and $d$ is the \emph{first changed} distance, and let $\mathcal{R}$ be any monotonic voting rule, then convergence is guaranteed, and the number of iterations is exactly $\max_{v \in V^0} \lceil d(v, w^0)/\epsilon\rceil$.
\end{corollary}

\section{Outlook}\label{section:outlook}

We introduced a model of iterative deliberation in metric spaces and instantiated it with several natural social choice settings, by selecting appropriate metric spaces and voting rules.
We identified those settings for which convergence of the process is guaranteed, and provided upper bounds regarding the number of iterative steps required for consensus (for those settings in which deliberation is guaranteed to succeed in finding a consensus).
Below we mention some directions for future research:
\begin{itemize}

\item
It is natural to consider further metric spaces, as well as further voting rules; a natural place to look for relevant metric spaces and voting rules is the work of Bulteau et al.~\cite{bulteau2021aggregation}.

\item
Another, more relaxed model, that comes to mind, is one where each voter must approach the winner by up to $\epsilon$, instead of exactly by $\epsilon$, and at least one voter must approach the winner by at least $\delta$.
This is a more general model, which is a bit closer to reality, where there is only an upper bound on the movement, and a demand that there is movement in each iteration. 
The demand is at least $\delta$ and not just larger than zero, because otherwise we could only reach convergence in the limit when the number of iterations approaches $\infty$.
\item
It is natural to study a stochastic model of iterative deliberation, including such that include radicalization (meaning, that an agent can move away from the aggregated point, instead of approaching it).
A stochastic model, in which such moves happen according to some probability or probability distribution may be closer to reality, thus has the potential of shedding more light on settings for which a deterministic process may converge, but a stochastic process may not.
For a stochastic model that incorporates a non-zero probability for radicalizing voters, intuitively, if the probability mass of radicalization is not too large, then convergence shall be maintained. 

\item
Another idea would be to consider coalition structures (such as those of Elkind et al.~\cite{elkind2021united}) in which the agents of each coalition move slightly towards the center of each coalition and study issues of convergence there. (This would be different than the one-coalition setting we consider here; in a way, this would be like several dynamic deliberation groups.)

\end{itemize}

\section*{Acknowledgement}

Nimrod Talmon and Eyal Leizerovich were supported by the Israel Science Foundation (ISF; Grant No. 630/19).

More importantly, we thank the \textbf{Hoodska Explosive} for years of fun.
\\
\\
\\

\bibliographystyle{splncs04}
\bibliography{bib}

\appendix



\end{document}